\documentclass[11pt]{article}
\usepackage{amsthm,amsfonts,amsmath,amssymb,times}
\usepackage{enumerate, graphicx}
\usepackage{algorithmic}
\usepackage{xspace}
\usepackage{url}
\usepackage{hyperref}
\usepackage{caption,subcaption}
\usepackage{tcolorbox}

\advance \topmargin by -1.0in
\advance \oddsidemargin by -0.8in
\newcommand{\myparskip}{3pt}
\parskip \myparskip

\newcommand{\headers}[3]{
\newpage\setcounter{page}{1}
\def\@oddhead{$\underline{\hbox to\textwidth{%
\textbf{\rlap{#1}\phantom{hj}\hfill #2 \hfill \llap{#3}}}}$}
\def\@oddfoot{\hfill\thepage\hfill}}

\newtheorem*{theorem*}{Theorem}

\newtheorem{lemma}{Lemma}[section]
\newtheorem{theorem}[lemma]{Theorem}

\newtheorem{remark}[lemma]{Remark}

\renewenvironment{proof}{\vspace{-0.1in}\noindent{\bf Proof:}}%
        {\hspace*{\fill}$\Box$\par}
        {\hspace*{\fill}$\Box$\par}
        {\hspace*{\fill}$\Box$\par}

\def\eps{\varepsilon}

\def\ceil#1{\lceil {#1} \rceil}
\def\script#1{\mathcal{#1}}
\def\opt{\textsc{OPT}}
\def\etal{\text{et al.}\xspace}

\def\mM{\script{M}}

\newcommand{\E}{\mathrm{E}}

 \newcommand{\cU}{\mathcal{U}}
\newcommand{\cS}{\mathcal{S}} \newcommand{\setcover}{{\sc Set
    Cover}\xspace}
\newcommand{\psc}{{\sc Partial-SC}\xspace}
\newcommand{\maxcover}{{\sc Max $k$-Cover}\xspace}
\newcommand{\maxbcover}{{\sc Max-Budgeted-Cover}\xspace}
\newcommand{\vc}{{\sc VC}\xspace}
\newcommand{\partitionsc}{{\sc Partition-SC}\xspace}
\newcommand{\cip}{{\sc CIP}\xspace}
\newcommand{\cips}{{\sc CIPs}\xspace}
\newcommand{\submodsc}{{\sc Submodular Set Cover}\xspace}
\newcommand{\multisubmodsc}{{\sc MP-Submod-SC}\xspace}
\newcommand{\ccf}{{\sc CCF}\xspace}
\newcommand{\sclp}{{\sc SC-LP}\xspace}
\newcommand{\psclp}{{\sc PSC-LP}\xspace}

\newcommand{\mbclp}{{\sc MBC-LP}\xspace}

\title{On Approximating Partial Set Cover and Generalizations}
\author{
  Chandra Chekuri\thanks{Dept.\ of Computer Science, University of
    Illinois, Urbana-Champaign, IL, 61820. {\tt chekuri@illinois.edu}.
    Work on this paper partially supported by NSF grant CCF-1526799.  }
  \and
  Kent Quanrud\thanks{Dept.\ of Computer Science, University of
    Illinois, Urbana-Champaign, IL, 61820. {\tt quanrud2@illinois.edu}.
    Work on this paper partially supported by NSF grant CCF-1526799.  }
  \and
  Zhao Zhang\thanks{College of Mathematics and Computer Science,
    Zhejiang Normal University, China.
     {\tt zhaozhang@zjnu.edu.cn}. Work on this paper partially supported by NSFC (11771013, 61751303, 11531011) and ZJ-NSFC (LD19A010001), and done while the author was visiting University of Illinois. }
}
\begin{document}
\maketitle

\begin{abstract}
  Partial Set Cover (\psc) is a generalization of the well-studied Set
  Cover problem (\setcover). In \psc the input consists of an integer
  $k$ and a set system $(\cU,\cS)$ where $\cU$ is a finite set, and
  $\cS \subseteq 2^{\cU}$ is a collection of subsets of $\cU$. The
  goal is to find a subcollection $\cS' \subseteq \cS$ of smallest
  cardinality such that sets in $\cS'$ cover at least $k$ elements of
  $\cU$; that is $|\cup_{A \in \cS'} A| \ge k$. \setcover is a special
  case of \psc when $k = |\cU|$.  In the weighted version each set
  $S \in \cS$ has a non-negative weight $w(S)$ and the goal is to find
  a minimum weight subcollection to cover $k$ elements. Approximation
  algorithms for \setcover have been adapted to obtain comparable
  algorithms for \psc in various interesting cases.  In recent work
  Inamdar and Varadarajan \cite{IV1}, motivated by geometric set
  systems, obtained a simple and elegant approach to reduce \psc to
  \setcover via the natural LP relaxation. They showed that if a
  deletion-closed family of \setcover admits a $\beta$-approximation
  via the natural LP relaxation, then one can obtain a
  $2(\beta + 1)$-approximation for \psc on the same family. In a
  subsequent paper \cite{IV2}, they also considered a generalization of
  \psc that has multiple partial covering constraints which is partly
  inspired by and generalizes previous work of Bera \etal
  \cite{BeraGKR} on the Vertex Cover problem.

  Our main goal in this paper is to demonstrate some useful
  connections between the results in \cite{BeraGKR,IV1,IV2} and
  submodularity. This allows us to simplify, and in some cases improve
  their results. We improve the approximation for \psc to
  $(1-1/e)(\beta + 1)$ in the same setting as that in \cite{IV1}.  We
  extend the results in \cite{BeraGKR,IV2} to the sparse setting.
\end{abstract}

\section{Introduction}
\label{sec:intro}
\setcover is a well-studied problem in combinatorial optimization. The
input is a set system $(\cU,\cS)$ consisting of a finite set $\cU$ and
a collection $\cS = \{S_1,S_2,\ldots,S_m\}$ of subsets of $\cU$. The
goal is to find a minimum cardinality subcollection
$\cS' \subseteq \cS$ such that $\cU$ is covered by sets in $\cS'$.  In
the weighted version each $S_i$ has a weight $w_i \ge 0$ and the goal
is to find a minimum weight subcollection of sets whose union is
$\cU$.  \setcover is NP-Hard and approximation algorithms have been
extensively studied. A very simple greedy algorithm yields an
$H_d \le 1 + \ln d$ approximation where $d = \max_i |S|_i$ and this
holds even in the weighted case. Moreover this bound is essentially
tight unless $P = NP$ \cite{Feige}. Various special cases of \setcover
have been studied in the literature. A well-known example is the
Vertex Cover problem in graphs (\vc) which can be viewed as a special
case of \setcover where the frequency of each element is at most $2$
(the frequency of an element is the number of sets it is contained
in). When the maximum frequency of is $f$, an
$f$-approximation can be obtained.  Interesting class of \setcover
instances come from various geometric range spaces in low
dimensions. A canonical example here is the problem of covering points
in the plane by a given collection of disks. This problem admits a
constant factor approximation in the weighted case \cite{ChanGKS} via
a natural LP, and a PTAS in the unweighted case \cite{MR} via local
search; there is also a QPTAS for the weighted case
\cite{MRR}. Closely related to \setcover are maximization variants,
namely, \maxcover and \maxbcover.  In \maxcover we are given a set
system $(\cU,\cS)$ and an integer $k$; the goal is to pick $k$ sets
from $\cS$ to maximize the size of their union. In \maxbcover the sets
have weights and the goal is to pick a collections of sets with total
weight at most a given budget $B$ so as to maximize the size of their
union.  $(1-1/e)$-approximations are known for both these problems
\cite{NWF,KMN,Svir} and there are tight unless $P=NP$ \cite{Feige}.

\paragraph{Partial Set Cover (\psc):} In \psc the input, in addition
to the set system as in \setcover, also has an integer parameter $k$,
and now the goal is to find a minimum (weight) subcollection of the
given sets whose union is of size at least $k$. Note that \setcover is
a special case when $k = |\cU|$. It is natural to ask if \psc can be
approximated (almost) as well as \setcover. In several settings this is
indeed the case. For instance the greedy algorithm gives the
same guarantee for \psc as it does for \setcover; one can see this
transparently by viewing \setcover and \psc as special case of the
\submodsc problem for which greedy has been analyzed by Wolsey
\cite{wolsey}. However, for special cases of \setcover such as \vc one needs
more careful analysis to obtain comparable bounds for \psc; we
refer the reader to \cite{KPS} and references therein.
Of particular interest to us is the recent result of Inamdar and
Varadarajan \cite{IV1} which gave a simple and intuitive reduction
from \psc to \setcover via the natural LP relaxation.  Their black box
reduction to \setcover is particularly useful in geometric settings.
Inamdar and Varadarajan show that if there is a $\beta$-approximation
for a deletion-closed class of \setcover instances\footnote{We say tha
  a family of set systems is \emph{deletion closed} if removing an
  element or removing a set from a set system in the family yields
  another set system in the same family.} via the standard LP, then
there is $2(\beta + 1)$ approximation for \psc on the same family,via
a standard LP relaxation.

In a subsequent paper Inamdar and Varadarajan \cite{IV2} considered a
generalization of \psc when there are multiple partial covering
constraints. They call their problem the Partition Set Cover problem
(\partitionsc) and were motivated by previous work of Bera \etal
\cite{BeraGKR} who considered the same problem in the special setting
of \vc. In this problem the input is a set system $(\cU, \cS)$, and
$r$ subsets $\cU_1,\cU_2,\ldots,\cU_r$ of $\cU$, and
$r$ integers $k_1,\ldots,k_r$. The goal is to find a minimum
cardinality subcollection $\cS' \subseteq \cS$ (or a minimum weight
subcollection in the weighted case) such that, for $1 \le i \le r$,
the number of elements covered by $\cS'$ from $\cU_i$ is at least
$k_i$. For deletion-closed set families that admit a
$\beta$-approximation for \setcover, \cite{IV2} obtained an
$O(\beta + \log r)$ approximation for \partitionsc and this
generalizes the results of \cite{BeraGKR,IV1}. In
\cite{hk18} the authors describe a primal-dual algorithm
that yields an $(fH_r + H_r)$-approximation for \partitionsc where
$f$ is the maximum frequency; note that \cite{IV2} implies
a ratio of $O(f + \log r)$ for the same problem
where the asymptotic notation hides a constant factor.

\paragraph{Submodular Set Cover and Related Problems:} As we remarked
\setcover and \psc are special cases of \submodsc. Given a finite
ground set $N$ a real-valued set function
$f:2^N \rightarrow \mathbb{R}$ is \emph{submodular} iff
$f(A\cup B) + f(A \cap B) \le f(A) + f(B)$ for all $A, B \subseteq
N$. A set function is monotone if $f(A) \le f(B)$ for all
$A \subset B$. We will be mainly interested here in monotone
submodular functions that are normalized, that is $f(\emptyset) = 0$,
and hence are also non-negative. A polymatroid is an integer valued
normalized monotone submodular function.  In \submodsc we are given
$N$, a non-negative weight function $w: N \rightarrow \mathbb{R}_+$,
and a polymatroid $f:2^N \rightarrow \mathbb{Z}_+$ via a value oracle.
The goal is to solve $\min_{S \subseteq N} w(S)$ such that
$f(S) = f(N)$. \setcover and \psc can be seen as special case of
\submodsc as follows. Given a set system $(\cU, \cS)$ let $N = [m]$
where $m = |\cS|$. Define the coverage function
$f:2^N \rightarrow \mathbb{R}_+$ as: $f(A) = |\cup_{i \in A} S_i|$. It
is well-known and easy to show that $f$ is a polymatroid.  Thus
\setcover can be reduced to \submodsc via the coverage function.  To
reduce \psc to \submodsc we let
$f(A) = \min\{k, |\cup_{i \in A} S_i|\}$ which we refer to as the
truncated coverage function.  Wolsey showed that a simple
greedy algorithm yields a $1 + \ln d$ approximation for \submodsc when
$f$ is a polymatroid where $d = \max_{i \in N} f(i)$.

Har-Peled and Jones \cite{HJ18}, motivated by an application from
computational geometry, implicitly considered the following
generalization of \submodsc. We have a ground set $N$ and a weight
function $w: N \rightarrow \mathbb{R}_+$ as before. Instead of one
polymatroid we are given $r$ polymatroids $f_1,f_2,\ldots,f_r$ over
$N$ and integers $k_1,k_2,\ldots,k_r$. The goal is to find
$S \subseteq N$ of minimum weight such that $f_i(S) \ge k_i$ for
$1 \le i \le r$. We refer to this as \multisubmodsc. As noted in
\cite{HJ18}, it is not hard to reduce \multisubmodsc to \submodsc. We
simply define a new function $g$ where
$g(A) = \sum_{i=1}^r \min\{k_i, f_i(A)\}$. Via Wolsey's result for
\submodsc this implies an $O(\log r + \log K)$ approximation via the
greedy algorithm where $K = \sum_{j=1}^r k_j$.
Although \multisubmodsc can be reduced to \submodsc it is useful
to treat it separately when the functions $f_i$ are not general
submodular functions, as is the case in \partitionsc.

We mention that prior work has considered multiple submodular
objectives from a \emph{maximization} perspective \cite{CVZ,CJV}
rather than from a minimum cost perspective. There are useful
connections between these two perspectives. Consider \submodsc. We
could recast the exact version of this problem as $\max f(S)$ subject
to the constraint $w(S) \le B$ where $B$ is a budget.
This is submodular function maximization subject to a knapsack
constraint and admits a $(1-1/e)$-approximation \cite{Svir}.

\paragraph{Covering Integer Programs (\cips):} A \cip is an integer
program of the form
$\max\{ wx \mid Ax \ge b, x \le d, x \in \mathbb{Z}^n_+\}$ where $A$
is a non-negative $m \times n$ matrix and $b \ge 0$.  \cip s
generalize \setcover and can be seen as a special case of \submodsc.
However, a direction reduction of \cip to \submodsc requires one to
scale the numbers and consequently the greedy algorithm does not yield
a good approximation ratio as a function of $m$ and $n$. This is
rectified via LP relaxations that employ knapsack-cover (KC)
inequalities first used in this context by Carr \etal \cite{CFLP}. Via
KC inequalities one obtains refined results for \cips that are
similar to those for \setcover modulo lower order terms. In particular
an $O(\log \Delta_0)$ approximation can be achieved where $\Delta_0$
is the maximum number of non-zeroes in any column of $A$.  We refer
the reader to \cite{KY,CHS,CQ19} for further results on \cips.

\subsection{Motivation and contributions}
Our initial motivation was to simplify and explain certain technical
aspects of the algorithm and analysis in \cite{IV1,IV2}.  We view \psc
and \partitionsc as special cases of \multisubmodsc and use the lens
of submodularity and bring in some known tools from this area.
This view point sheds light on the properties of the
coverage function that lead to stronger bounds than those possible for
general submodular functions. A second perspective we bring is from
the recent work on \cips \cite{CQ19} that shows the utility of a
simple randomized-rounding plus alteration approach to obtain
approximation ratios that depend on the sparsity.  Using these two
perspectives we obtain some improvements and generalizations of the
results in \cite{IV1,IV2}.

\begin{itemize}
\item For deletion-closed set systems that have a
  $\beta$-approximation to \setcover via the natural LP we obtain
  a $(1-1/e)(\beta + 1)$-approximation for \psc. This slightly
  improves the bound of \cite{IV1} from $2(\beta +1)$ while also
  simplifying the algorithm and analysis.
\item For \multisubmodsc we obtain a bicriteria approximation.
  We obtain a random solution $S$ such that $f_i(S) \ge (1-1/e-\eps) k_i$
  for $1 \le i \le r$ and the expected weight of $S$ is
  $O(\frac{1}{\eps}\log r) \opt$. We obtain the same bound even in a more general
  setting where the system of constraints is $r$-sparse. We describe
  an application of the bicriteria approximation to splitting point
  sets that was considered in \cite{HJ18}.
\item We consider a simultaneous generalization of \partitionsc and
  \cips and obtain a randomized $O(\beta + \log r)$ approximation
  where $r$ is the sparsity of the system. This generalizes the
  result of \cite{IV2} to the sparse setting.
\end{itemize}

We hope that some of the ideas here are useful in extending work on
\psc and generalizations to other special cases of submodular functions.

\section{Background}
\label{sec:prelim}
\setcover and \psc have natural LP relaxations and they are closely
related to those for \maxcover and \maxbcover. The LP relaxation
for \setcover (\sclp) is shown in Fig~\ref{fig:sclp}. It has
a variable $x_i$ for each set $S_i \in \cS$, which, in the integer
programming formulation, indicates whether $S_i$ is picked in the
solution. The goal is to minimize the weight of the chosen sets
which is captured by the objective $\sum_{S_i \in \cS} w_i x_i$
subject to the constraint that each element $e_j$ is covered.
The LP relaxation for \psc (\psclp) is shown in Fig~\ref{fig:psclp}.
Now we need additional variables to indicate which of the $k$
elements are going to be covered; for each $e_j \in \cU$ we thus
have a variable $z_j$ for this purpose. In \psclp it is important
to constrain $z_j$ to be at most $1$. The constraint $\sum_{e_j} z_j
\ge k$ forces at least $k$ elements to be covered fractionally.

\begin{figure}[htb]
  \begin{subfigure}{0.5\linewidth}
  \begin{center}
    \begin{tcolorbox}[width=3in]
      (SC-LP)
      \begin{align*}
        \min \sum_{S_i \in \cS} w_i x_i & \\
        \sum_{i: e_j \in S_i} x_i & \ge 1 \quad e_j \in \cU\\
        x_i & \ge 0 \quad \quad S_i \in \cS
      \end{align*}
    \end{tcolorbox}
  \end{center}
  \caption{LP relaxation for \setcover. }
  \label{fig:sclp}
  \end{subfigure}
  \begin{subfigure}{0.5\linewidth}
 \begin{center}
   \begin{tcolorbox}[width=3in]
     (PSC-LP)
      \begin{align*}
        \min \sum_{S_i \in \cS} w_i x_i & \\
        \sum_{i: e_j \in S_i} x_i & \ge z_j \quad e_j \in \cU\\
        \sum_{e_j \in \cU} z_j & \ge k \\
        z_j & \in [0,1] \quad e_j \in \cU \\
        x_i & \ge 0 \quad \quad S_i \in \cS
      \end{align*}
    \end{tcolorbox}
  \end{center}
  \caption{LP relaxation for \psc. }
  \label{fig:psclp}
  \end{subfigure}
\end{figure}

As noted in prior work the integrality gap of \psclp can be made
arbitrarily large but it is easy to fix by guessing the largest
cost set in an optimum solution and doing some preprocessing.
We discuss this issue in later sections.

Figs~\ref{fig:mclp} and \ref{fig:mbclp} show LP relaxations
for \maxcover and \maxbcover respectively. In these problems
we maximize the number of elements covered subject to an upper bound
on the number of sets or on the total weight of the chosen sets.

\begin{figure}[htb]
  \begin{subfigure}{0.5\textwidth}
      \begin{center}
        \begin{tcolorbox}[width=3in]
          (MC-LP)
          \begin{align*}
            \max \sum_{e_j \in \cU} z_j & \\
            \sum_{i: e_j \in S_i} x_i & \ge z_j \quad e_j \in \cU\\
            \sum_{S_i \in \cS} x_i & \le k \\
            z_j & \in [0,1] \quad e_j \in \cU \\
            x_i & \ge 0 \quad \quad S_i \in \cS
          \end{align*}
        \end{tcolorbox}
      \end{center}
      \caption{LP relaxation for \maxcover. }
      \label{fig:mclp}
    \end{subfigure}
    \begin{subfigure}{0.5\textwidth}
      \begin{center}
        \begin{tcolorbox}[width=3in]
          (MBC-LP)
          \begin{align*}
            \max \sum_{e_j \in \cU} z_j & \\
            \sum_{i: e_j \in S_i} x_i & \ge z_j \quad e_j \in \cU\\
            \sum_{S_i \in \cS} w_i x_i & \le B \\
            z_j & \in [0,1] \quad e_j \in \cU \\
            x_i & \ge 0 \quad \quad S_i \in \cS
          \end{align*}
        \end{tcolorbox}
      \end{center}
      \caption{LP relaxation for \maxbcover. }
      \label{fig:mbclp}
    \end{subfigure}
\end{figure}

\paragraph{Greedy algorithm:} The greedy algorithm is a well-known and
standard algorithm for the problems studied here. The algorithm
iteratively picks the set with the current maximum bang-per-buck ratio
and add it to the current solution until some stopping condition is
met. The bang-per-buck of a set $S_i$ is defined as
$|S_i \cap \cU'|/w_i$ where $\cU'$ is the set of uncovered elements at
that point in the algorithm. For minimization problems such as
\setcover and \psc the algorithm is stopped when the required number
of elements are covered. For \maxcover and \maxbcover the algorithm is
stopped when if adding the current set would exceed the budget. Since
this is a standard algorithm that is extremely well-studied we do not
describe all the formal details and the known results. Typically the
approximation guarantee of Greedy is analyzed with respect to an
optimum integer solution. We need to compare it to the value of the
fractional solution. For the setting of the cardinality constraint
this was already done in \cite{NWF}.  We need a slight generalization
to the budgeted setting and we give a proof for the sake of
completeness.

\begin{lemma}
  \label{lem:greedy}
  Let $Z$ be the optimum value of (\mbclp) for a given instance of
  \maxbcover with budget $B$.
  \begin{itemize}
  \item  Suppose Greedy algorithm is run
    until the total weight of the chosen sets is
    equal to or exceeds $B$. Then the number of elements covered by
    greedy is at least $(1-1/e)Z$.
  \item Suppose no set covers more than $c Z$ elements for some $c > 0$
    then the weight of sets chosen by Greedy to cover $(1-1/e)Z$
    elements is at most $(1+ec)B$.
  \end{itemize}
  These conclusions holds even for the weighted coverage problem.
\end{lemma}
\begin{proof}
  We give a short sketch. Greedy's analysis for \maxbcover is based on
  the following key observation. Consider the first set $S$ picked by
  Greedy.  Then $|S|/w(S) \ge \opt/B$ where $\opt$ is the value of an
  optimum integer solution. And this follows from submodularity of the
  coverage function. This observation is applied iteratively with the
  residual solution as sets are picked and a standard analysis shows
  that when Greedy first meets or exceeds the budget $B$ then the
  total number of elements covered is at least $(1-1/e)\opt$.  We
  claim that we can replace $\opt$ in the analysis by $Z$.  Given a
  fractional solution $x,z$ we see that
  $Z = \sum_e z_e \le \sum_{e \in \cU} \min\{1,\sum_{i: e \in S_i}
  x_i\}$. Moreover $\sum_i w_i x_i \le B$. Via simple algebra, we can
  obtain a contradiction if $|S_i|/w_i < Z/B$ holds for all sets
  $S_i$.  Once we have this property the rest of the analysis is very
  similar to the standard one where $\opt$ is replaced by $Z$.

  Now consider the case when no set covers more than $cZ$ elements.
  If Greedy covers $(1-1/e)Z$ elements before the weight of sets
  chosen exceeds $B$ then there is nothing to prove. Otherwise let
  $S_j$ be the set added by Greedy when its weight exceeds $B$ for the
  first time. Let $\alpha \le |S_j|$ be the number of new elements
  covered by the inclusion of $S_j$. Since Greedy had covered less
  than $(1-1/e)Z$ elements the value of the residual fractional
  solution is at least $Z/e$. From the same argument as the in
  the preceding paragraph, since Greedy chose $S_j$ at that point,
  $\frac{\alpha}{w(S_j)} \ge \frac{Z}{eB}$. This implies that
  $w(S_j) \le e B \frac{\alpha}{Z} \le e c B$. Since Greedy covers
  at least $(1-1/e)Z$ elements after choosing $S_j$ (follows from
  the first claim of the lemma), the total weight of the sets
  chosen by Greedy is at most $B + w(S_j) \le (1+ec)B$.
\end{proof}

\subsection{Submodular set functions and continuous extensions}

Continuous extensions of submodular set functions have played
an important role in algorithmic and structural aspects.
The idea is to extend a discrete set function $f:2^N \rightarrow
\mathbb{R}$ to the continous space $[0,1]^N$.
Here we are mainly concerned with extensions motivated by
maximization problems, and confine our attention to two
extensions and refer the interested reader to
\cite{CCPV07,Vondrak-thesis} for a more detailed discussion.

The \emph{multilinear extension} of a real-valued
set function $f: 2^N \rightarrow \mathbb{R}$,
denoted by $F$, is defined as follows: For $x \in [0,1]^N$
$$F(x) = \sum_{S \subseteq N} f(S) \prod_{i \in S} x_i \prod_{j \not
  \in S} (1-x_j).$$
Equivalently $F(x) = \E[f(R)]$ where $R$ is a random set obtained
by picking each $i \in N$ independently with probability $x_i$.

The \emph{concave closure} of a real-valued set function
$f: 2^N \rightarrow \mathbb{R}$, denoted by $f^+$, is defined as
the optimum of an exponential sized linear program:
$$f^+(x) = \max \sum_{S \subseteq N} f(S) \alpha_S \quad \text{~s.t~} \sum_{S \ni
  i} \alpha_S = x_i \quad \forall i \in N \text{~and~} \alpha_S \ge 0
\quad \forall S.$$

A special case of submodular functions are non-negative weighted sums
of rank functions of matroids. More formally suppose $N$ is a finite
ground set and $\mM_1,\mM_2,\ldots,\mM_\ell$ are $\ell$ matroids on
the same ground set $N$. Let $g_1, \ldots,g_\ell$ be the rank
functions of the matroids and these are monotone submodular. Suppose
$f = \sum_{h=1}^\ell w_h g_h$ where $w_h \ge 0$ for all
$h \in [\ell]$, then $f$ is monotone submodular. We note that
(weighted) coverage functions belongs to this class.
For a such a submodular function we can consider an
extension $\tilde{f}$ where $\tilde{f}(x) = \sum_h w_h g^{+}(x)$.
We capture two useful facts which are shown in \cite{CCPV07}.

\begin{lemma}[\cite{CCPV07}]
  \label{lem:extensions}
  Suppose $f = \sum_{h=1}^\ell w_h g_h$ is the weighted sum of rank
  functions of matroids. Then $F(x) \ge (1-1/e)\tilde{f}(x)$.
  Assuming oracle access to the rank functions $g_1,\ldots,g_\ell$,
  for any $x \in [0,1]^N$, there is a polynomial-time solvable LP
  whose optimum value is $\tilde{f}(x)$.
\end{lemma}

\begin{remark}
  Let $f: 2^{\cS} \rightarrow \mathbb{Z}_+$ be the coverage function
  associated with a set system $(\cU, \cS)$. Then
  $\tilde{f}(x) = \sum_{e \in \cU} \min\{1, \sum_{i: e \in S_i} x_i\}$
  where $\tilde f = \sum_{e \in \cU} g_e^+$ and $g^{+}_e(x) =  \min\{1,
  \sum_{i: e \in S_i} x_i\}$ is the rank function of a simple uniform
  matroid. One can see \psclp in a more compact fashion:
  $$\min \sum_i w_i x_i \text{~~s.t~~} \tilde{f}(x) \ge k.$$
\end{remark}

\paragraph{Concentration under randomized rounding:}
Recall the multilinear extension $F$ of a submodular
function $f$. If $x \in [0,1]^N$ then $F(x) = \E[f(R)]$ where
$R$ is a random set obtained by independently including
each $i \in N$ in $R$ with probability $x_i$.
We can ask whether $f(R)$ is concentrated around $\E[f(R)] = F(x)$.
And indeed this is the case when $f$ is Lipscitz.
For a parameter $c \ge 0$, $f$ is $c$-Lipschitz if
$|f_A(i)| \le c$ for all $i \in N$ and $A \subset N$; for
monotone functions this is equivalent to
the condition that $f(i) \le c$ for all $i \in N$.

\begin{lemma}[\cite{submod-chernoff}]
  \label{lem:submod-concentration}
  Let $f:2^N\rightarrow \mathbb{R}_+$ be a $1$-Lipschitz monotone
  submodular function. For $x \in [0,1]^N$ let $R$ be a random set
  drawn from the product distribution induced by $x$. Then for
  $\delta \ge 0$,
  \begin{itemize}
  \item $\Pr[f(R) \ge (1+\delta) F(x)] \le (\frac{e^\delta}{(1+\delta)^{(1+\delta)}})^{F(x)}$.
  \item $\Pr[f(R) \le (1-\delta) F(x)] \le e^{-\delta^2 F(x)/2}$.
  \end{itemize}
\end{lemma}

\paragraph{Greedy algorithm under a knapsack constraint:}
Consider the problem of maximizing a monotone submodular function
subject to a knapsack constraint; formally
$\max f(S) \text{~s.t~} w(S) \le B$ where
$w: N \rightarrow \mathbb{R}_+$ is a non-negative weight function on
the elements of the ground set $N$. Note that when all $w(i) = 1$ and
$B = k$ this is the problem of maximizing a monotone submodular
function subject to a cardinality constraint. For the cardinality
constraint case, the simple Greedy algorithm that iteratively picks
the element with the largest marginal value yields a
$(1-1/e)$-approximation \cite{NWF}.  Greedy extends in a
natural fashion to the knapsack constraint setting; in each iteration
the element $i = \arg\max_{j} f_S(j)/w_j$ is chosen where $S$ is the
set of already chosen elements. Sviridenko \cite{Svir}, building on
earlier work on the coverage function \cite{KMN}, showed that Greedy
with some partial enumeration yields a $(1-1/e)$-approximation
for the knapsack constraint. The following lemma quantifies the
performance of the basic Greedy when it is stopped after meeting or
exceeding the budget $B$.

\begin{lemma}
  \label{lem:greedysubmod}
  Consider an instance of monotone submodular function maximization
  subject to a knapsack constraint.  Let $Z$ be the optimum value for
  the given knapsack budget $B$. Suppose the greedy algorithm is run
  until the total weight of the chosen sets is equal to or exceeds
  $B$. Letting $S$ be the greedy solution we have
  $f(S) \ge (1-1/e)Z$.
\end{lemma}

\section{Approximating \psc}
\label{sec:psc}
In this section we consider the algorithm for \psc from \cite{IV1}
and suggest a small variation that simplifies the algorithm and
analysis. The approach of \cite{IV1} is as follows. Given
an instance of \psc with a set system $(\cU,\cS)$ their algorithm
has the following high level steps.
\begin{enumerate}
\item Guess the largest weight set in an optimum solution. Remove all
  elements covered by it, remove all sets with weight larger than the
  guessed set. Adjust $k$ to account for covered elements. We now work
  with the residual instance of \psc.
\item Solve \psclp. Let $(x^*,z^*)$ be an optimum solution.
  For some threshold $\tau$ let
  $\cU_h = \{ e_j \mid z^*_j \ge \tau\}$ be the highly covered elements
  and let $\cU_\ell = \{e_j \mid z^*_j < \tau\}$ be shallow elements.
\item Solve a \setcover instance
  via the LP to cover all elements in $\cU_h$. The cost of this
  solution is at most $\frac{1}{\tau}\beta \sum_i w_i x_i^*$ since
  one can argue that the fractional solution $x'$ where $x'_i =
  \min\{1,x^*_i/\tau\}$ for each $i$ is a feasible
  fractional solution for \sclp to cover $\cU_h$.
\item Let $k'  = k - |\cU_h|$ be the residual number of elements
  that need to be covered from $\cU_\ell$. Round
  $(x^*,z^*)$ to cover $k'$ elements from $\cU_\ell$.
\end{enumerate}

The last step of the algorithm is the main technical one, and also
determines $\tau$. In \cite{IV1} $\tau$ is chosen to be $1/2$ and this
leads to their $2(\beta + 1)$-approximation. The rounding algorithm in
\cite{IV1} can be seen as an adaptation of pipage rounding
\cite{AgeevS} for \maxbcover. The details are somewhat technical and
perhaps obscure the high-level intuition that scaling up the LP
solution allows one to use a bicriteria approximation for \maxbcover.
Our contribution is to simplify the fourth step in the preceding
algorithm. Here is the last step in our algorithm; the other steps are
the same modulo the specific choice of $\tau$.

\begin{enumerate}
\item[4'.] Run Greedy to cover $k'$ elements from $\cU_\ell$.
\end{enumerate}

We now analyze the performance of our modified algorithm.

\begin{lemma}
  \label{lem:main}
  Suppose $\tau \le (1-1/e)$. Then running Greedy in the final step
  outputs a solution of total weight at most
  $\max_i w_i + \frac{1}{\tau}\sum_i w_i x_i^*$ to cover
  $k' = k - |\cU_h|$ elements from $\cU_\ell$.
\end{lemma}
\begin{proof}
  It is easy to see that $\sum_{e_j \in \cU_\ell} z^*_j \ge k'$
  since $\sum_{e_j \in \cU} z^*_j \ge k$ and $z^*_j \le 1$ for each
  $e_j$.  Let $(\cU_\ell, \cS')$ be the set system obtained by
  restricting $(\cU, \cS)$ to $\cU_\ell$,  and let $(x',z')$ be the
  restriction of $(x^*,z^*)$ to the set system $(\cU_\ell, \cS')$. We
  have (i) $\sum_i w_i x'_i \le \sum_i w_i x^*_i$ and (ii)
  $\sum_{e_j \in \cU_\ell} z'_j \ge k'$ and (iii)
  $z'_j \le \tau \le (1-1/e)$ for all $e_j \in \cU_\ell$.

  Consider $(x'',z'')$ obtained from $(x',z')$ as follows.  For each
  $e_j \in \cU_\ell$ set $z''_j = \frac{1}{\tau} z'_j$ and note that
  $z''_j \le 1$. For each set $S_i$ set
  $x''_i = \min\{1,\frac{1}{\tau} x'_i\}$. It is easy to see that
  $(x'',z'')$ is a feasible solution to \psclp. Note that
  $Z = \sum_{e_j \in \cU_\ell} z''_j \ge \frac{1}{\tau} k'$. Let
  $B = \sum_i w_i x''_i \le \frac{1}{\tau}\sum_i w_i x^*_i$. The
  fractional solution $(x'',z'')$ is also a feasible solution to the
  LP formulation \mbclp. We apply Lemma~\ref{lem:greedy} to this
  fractional solution. Suppose we stop Greedy when it covers $k'$
  elements or when it first crosses the budget $B$, whichever comes first. Clearly the total weight is at most
  $B + \max_i w_i$.  We argue that at least $k'$ elements are covered
  when we stop Greedy. The only case to argue is when Greedy is
  stopped when the weight of sets picked by it exceeds $B$ for the
  first time. From Lemma~\ref{lem:greedy} it follows that Greedy
  covers at least $(1-1/e)Z$ elements but since
  $Z \ge \frac{1}{\tau} k'$ it implies that Greedy covers at least
  $k'$ elements when it is stopped.
\end{proof}

We formally state a lemma to bound the cost of covering $\cU_h$.
We sketch the simple proof for the sake of completeness, it is
identical to that from \cite{IV1}.
\begin{lemma}
  \label{lem:highlycovered}
  The cost of covering $\cU_h$ is at most
  $\beta \frac{1}{\tau}\sum_i w_i x^*_i$.
\end{lemma}
\begin{proof}
    Recall that $z^*_j \ge \tau$ for each $e_j \in \cU_h$. Consider
  $x'_i = \min\{1,\frac{1}{\tau}x^*_i\}$. It is easy to see that $x'$ is
  a feasible fractional solution for \sclp to cover $\cU_h$ using sets
  in $\cS$. Since the set family is deletion-closed, and the integrality
  gap of the \sclp is at most $\beta$ for all instances in the family,
  there is an integral solution covering $\cU_h$ of cost at most
  $\beta \sum_i w_i x'_i \le \frac{1}{\tau}\beta \sum_i w_i x^*_i$.
\end{proof}

\begin{theorem}
  Setting $\tau = (1-1/e)$, the algorithm outputs a feasible solution
  of total cost at most $(1-1/e)(\beta + 1) \opt$ where $\opt$ is the
  value of an optimum integral solution.
\end{theorem}
\begin{proof}
  Fix an optimum solution.
  Let $W$ be the weight of a maximum weight set in the optimum
  solution. In the first step of the algorithm we can assume that the
  algorithm has correctly guessed a maximum weight set from the fixed
  optimum solution. Let $\opt' = \opt - W$. In the residual instance
  the weight of every set is at most $W$. The optimum solution value
  for \psclp, after guessing the largest weight set and removing it,
  is at most $\opt'$. From Lemma~\ref{lem:highlycovered}, the cost of
  covering $\cU_h$ is at most $\frac{e}{e-1}\beta \opt'$.  From
  Lemma~\ref{lem:main}, the cost of covering $k'$ elements from
  $\cU_\ell$ is most $\frac{e}{e-1}\opt' + W$. Hence the total cost,
  including the weight of the guessed set, is at most
  $$W + \frac{e}{e-1}\beta \opt' + \frac{e}{e-1}\opt' + W =  \frac{e}{e-1}(\beta + 1)
  \opt + W (2 - \frac{e}{e-1}(\beta+1)) \le \frac{e}{e-1}(\beta + 1)
  \opt$$
  since $\beta \ge 1$.
\end{proof}

\section{A bicriteria approximation for \multisubmodsc}
\label{sec:bicriteria}
In this section we consider \multisubmodsc. Let $N$ be a finite ground
set. For each $j \in [h]$ we are given a submodular function $f_j:2^{N}\rightarrow
\mathbb{R}_+$. We are also given a non-negative weight function $w: N
\rightarrow \mathbb{R}_+$. The goal is to solve the following
covering problem:
\begin{eqnarray*}
  \min_{S \subseteq N} w(S)  & & \text{s.t}\\
   f_j(S) & \ge & 1 \quad 1 \le j \le h
\end{eqnarray*}
We say that $i \in N$ is active in constraint $j$ if
$f_j(i) > 0$, otherwise it is inactive.
We say that the given instance is
$r$-sparse if each element $i \in N$ is active in at most $r$
constraints.

\begin{theorem}
  There is a randomized polynomial-time approximation algorithm that
  given an $r$-sparse instance of \multisubmodsc outputs a set $S
  \subseteq N$ such that (i) $f_j(S) \ge (1-1/e-\eps)$ for $1 \le j \le h$, and
  (ii) $E[w(S)] = O(\frac{1}{\eps}\ln r) \opt$.
\end{theorem}

The rest of the section is devoted to the proof of the preceding
theorem.  We will assume without loss of generality that for each $i$,
$f_i(N) \le 1$; otherwise we can work with the truncated function
$\min \{1, f_i(S)\}$ which is also submodular. This technical
assumption plays a role in the analysis later.

We consider a continuous relaxation of the problem based on the
multilinear extension. Instead of finding a set $S$ we consider
finding a fractional point $x \in [0,1]^N$. For any value $B \ge \opt$
where $\opt$ is the optimum value of the original problem, the
following continuous optimization problem has a feasible solution.

\begin{eqnarray*}
  \text{(MP-Submod-Relax)} \quad \sum_i w_i x_i & \le & B \\
  F_j(x) & \ge & 1 \quad 1 \le j \le h \\
                 x \ge 0
\end{eqnarray*}

One cannot hope to solve the preceding continuous optimization problem
since it is NP-Hard. However the following approximation result is
known and is based on extending the continuous greedy algorithm
of Vondrak \cite{V08,CCPV}.

\begin{theorem}[\cite{CVZ,CJV}]
  There is a randomized polynomial-time algorithm that given an
  instance of MP-Submod-Relax and value oracle access to the
  submodular functions $f_1,\ldots,f_h$, with high probability, either
  correctly outputs that the instance is not feasible or outputs an
  $x$ such that (i) $\sum_i w_i x_i \le B$ and (ii)
  $F_i(x) \ge (1-1/e - \eps)$ for $1 \le i \le h$.
\end{theorem}

Using the preceding theorem and binary search one can obtain an $x$
such that $\sum_{i\in N} w_i x_i \le \opt$ and
$F_j(x) \ge (1-1/e - \eps)$ for $1 \le j \le h$.  It
remains to round this solution. We use the following algorithm
based on the high-level framework of randomized rounding plus
alteration.

\begin{enumerate}
\item Let $S_1, S_2,\ldots,S_\ell$ be random sets obtained by picking elements
  independently and randomly $\ell$ times according to the fractional solution
  $x$. Let $S = \cup_{k=1}^\ell S_k$.
\item For each $j \in [h]$ if $f(S) < (1-1/e-2\eps)$,
  fix the constraint. That is, find a set $T_j$ using the greedy algorithm
  (via Lemma~\ref{lem:greedysubmod}) such that $f_j(T_j) \ge (1-1/e)$.
  We implicitly set $T_j = \emptyset$ if $f(S) \ge (1-1/e - 2\eps)$.
\item Output $S \cup T$ where $T = \cup_{j=1}^h T_j$.
\end{enumerate}

It is easy to see that $S \cup T$ satisfies the property that
$f_j(S \cup T) \ge (1-1/e - 2\eps)$ for $ j \in [h]$.
It remains to choose $\ell$ and bound the expected cost of $S \cup T$.

The following is easy from randomized rounding stage of the algorithm.
\begin{lemma}
  $\E[w(S)] = \ell \sum_{i=1}^h w_i x_i \le \ell \opt$.
\end{lemma}

We now bound the probability that any fixed constraint is not
satisfied after the randomized rounding stage of the algorithm.
Let $I_j$ be the indicator for the event that $f_j(S) < (1-1/e-2\eps)$.

\begin{lemma}
  For any $j \in [h]$, $\Pr[I_j] \le \alpha^{-\ell}$,
    where $\alpha \le 1 - \eps$ for sufficiently small $\eps > 0$.
\end{lemma}
\begin{proof}
  Let $I_{j,k}$ be indicator for the event that
    $f_j(S_k) < (1-1/e-2\eps)$.  From the definition of the
  multilinear extension, for any $k \in [\ell]$,
  $\E[f_j(S_k)] = F_j(x)$. Hence, $\E[f_j(S_k)] \ge (1-1/e -\eps)$.
  Let $\alpha = \Pr[I_{j,k}]$. We upper bound $\alpha$ as follows.
  Recall that $f_j(N) \le 1$ and hence by monotonicity we have
  $f_j(A) \le 1$ for all $A \subseteq N$. Since $\E[f_j(S_k)] \ge
  (1-1/e-\eps)$ we can upper bound $\alpha$ by the following:
  $$\alpha (1-1/e-2\eps) + (1-\alpha) \ge (1-1/e - \eps).$$
  Rearranging we have
  $\alpha \le \frac{(1 + e\eps)}{(1 + 2e\eps)} = \frac{1}{1+
    \frac{e\eps}{1+e\eps}}$. Using the fact that for
  $\frac{1}{1+x} \le 1- x/2$ for sufficiently small $x > 0$, we
  simplify and see that
  $\alpha \le 1- \frac{e \eps}{2(1+e \eps)} \le 1 - \eps$ for sufficiently
  small $\eps > 0$.  Since the
  sets $S_1,\ldots,S_\ell$ are chosen independently,
  $$\Pr[I_j] = \prod_{k = 1}^\ell \Pr[I_{j,k}] \le \alpha^{-\ell}.$$
\end{proof}

\begin{remark}
  The simplicity of the previous proof is based on the use of the
  multilinear extension which is well-suited for randomized rounding. The
  assumption that $f_j(N) \le 1$ is technically important and it is
  easy to ensure in the general submodular case but is not
  straightforward when working with specific classes of functions.
\end{remark}

\begin{lemma}
  \label{lem:sumcosts}
  Let $\opt_j$ be the value of an optimum solution to the problem
  $\min w(S) \text{~s.t~} f_j(S) \ge 1$. Then, $\sum_{j=1}^h \opt_j \le
  r \opt$.
\end{lemma}
\begin{proof}
  Let $S^*$ be an optimum solution to the problem of covering all $h$
  constraints. Let $N_j$ be the set of active elements for constraint
  $j$. It follows that $S^* \cap N_j$ is a feasible solution for the
  problem of covering just $f_j$. Thus $\opt_j \le w(S^* \cap N_j)$.
  Hence
  $$\sum_j \opt_j \le \sum_j w(S^* \cap N_j) = \sum_{i \in S^*} w_i
  \sum_{j:i \in N_j} 1 \le r w(S^*) = r \opt.$$
\end{proof}

We now bound the expected cost of $T$
\begin{lemma}
  $\E[w(T)] \le 2\alpha^{-\ell} \sum_j \opt_j \le 2 \alpha^{-\ell} r \opt$.
\end{lemma}
\begin{proof}
  We claim that $w(T_j) \le 2 \opt_j$. Assuming the claim, from
  the description of the algorithm, we have
  $$\E[w(T)] \le \sum_{j=1}^h \Pr[I_j] w(T_j) \le 2\alpha^{-\ell} \sum_j
  \opt_j \le 2\alpha^{-\ell} r \opt.$$
  Now we prove the claim.  Consider the
  problem $\min w(S) \text{~s.t~} f_j(S) \ge 1$.  $\opt_j$ is the
  optimum solution value to this problem.  Now consider the following
  submodular function maximization problem subject to a knapsack
  constraint: $\max f_j(S) \text{~s.t~} w(S) \le \opt_j$. Clearly the
  optimum value of this maximization problem is at least $1$. From
  Lemma~\ref{lem:greedysubmod}, the greedy algorithm when run on the
  maximization problem, outputs a solution $T_j$ such that
  $f(T_j) \ge (1-1/e)$ and $w(T_j) \le \opt_j + \max_i w_i$. By
  guessing the maximum weight element in an optimum solution to the
  maximization problem we can ensure that $\max_i w_i \le
  \opt_j$. Thus, $w(T_j) \le 2 \opt_j$ and $f(T_j) \ge (1-1/e)$.
\end{proof}

From the preceding lemmas it follows that
$$\E[w(S \cup T)] \le \E[w(S)] + \E[w(T)] \le \ell \opt + 2 \alpha^{-\ell}
r \opt.$$
We set $\ell = \ceil{\log_{\alpha}r} = O(\frac{1}{\eps} \ln r)$
one can see that $\E[w(S \cup T)] \le O(\frac{1}{\eps}\ln r) \opt.$

\subsection{An application to splitting point sets}
Har-Peled and Jones \cite{HJ18}, as we remarked, were motivated to
study \multisubmodsc due a geometric application. Their problem is the
following.  Given $m$ point sets $P_1,\ldots,P_m$ in $\mathbb{R}^d$
they wish to find the smallest number of hyperplanes (or other
geometric shapes) such that no point set $P_i$ has more than a
constant factor of its points in any cell of the arrangement induced
by the chosen hyperplanes; in particular when the constant is a half,
the problem is related to the Ham-Sandwich theorem which implies that
when $m \le d$ just one hyperplane suffices!\footnote{A polynomial time
algorithm to find such a hyperplane is not known however.}
From this one can infer that $\ceil{m/d}$ hyperplanes always suffice. Let $k_i = |P_i|$
and let $P = \cup_i P_i$. We will assume, for notational simplicity,
that the sets $P_i$ are disjoint. The assumption can be dispensed with.
We refer the reader to \cite{HJ18} for connections to Ham-Sandwich
theorem and other problems.

In \cite{HJ18} the authors reduce their problem to \multisubmodsc as
follows. Let $N$ be the set of all hyperplanes in $\mathbb{R}^d$; we
can confine attention to a finite subset by restricting to those
half-spaces that are supported by $d$ points of $P$. For each point
set $P_i$ they consider a complete graph $G_i$ on the vertex set
$P_i$. For each $p \in \cup_i P_i$ they define a submodular function
$f_p:2^N \rightarrow \mathbb{R}_+$ where $f_p(S)$ is the number of
edges incident to $p$ that are cut by $S$; an edge $(p,q)$ with
$p,q \in P_i$ is cut if $p$ and $q$ are separated by at least one of
the hyperplanes in $S$. Thus one can formulate the original problem as
choosing the smallest number of hyperplanes such that for each
$p \in P$ the number of edges that are cut is at least $k_p$ where $k_p$ is the
demand of $p$. To ensure that $P_i$ is partitioned such that no cell
has more than $k_i/2$ points we set $k_p = k_i/2$ for each
$p \in P_i$; more generally if we wish no cell to have more than
$\beta k_i$ points of $P_i$ we set $k_p = (1-\beta)k_i$ for each
$p \in P_i$. As a special case of \multisubmodsc we have
\begin{eqnarray*}
  \min_{S \subseteq N} |S|  & & \text{s.t}\\
   f_p(S) & \ge &  k_p \quad p \in P
\end{eqnarray*}
Using Wolsey's result for \submodsc, \cite{HJ18} obtain an $O(\log (mn))$
approximation where $n = \sum_i k_i$.

We now show that one can obtain an $O(\log m)$-approximation if we
settle for a bicriteria approximation where we compare the cost of the
solution to that of  an optimum solution, but guarantees a slightly weaker bound
on the partition quality. This could be useful since one can imagine
several applications where $m$, the number of different point sets, is
much smaller than the total number of points.  Consider the
formulation from \cite{HJ18}. Suppose we used our bicriteria
approximation algorithm for \multisubmodsc. The algorithm would cut
$(1-1/e-\eps)k_p$ edges for each $p$ and hence for $1 \le i \le m$
we will only be
guaranteed that each cell in the arrangement contains at most
$(1 - (1-1/e-\eps)/2)k_i$ points from $P_i$. This is acceptable
in many applications. However, the approximation ratio still depends
on $n$ since the number of constraints in the formulation is $n$.  We
describe a related but slightly modified formulation to obtain an
$O(\log m)$-approximation by using only $m$ constraints.

Given a collection $S \subseteq N$ let $f_i(S)$ denote the number of
pairs of points in $P_i$ that are separated by $S$ (equivalently the
number of edges of $G_i$ cut by $S$). It is easy to see that $f_i(S)$
is a monotone submodular function over $N$. Suppose $S \subseteq N$
induces an arrangement such that no cell in the arrangement contains
more than $(1-\beta)k_i$ points for some $0 < \beta < 1$.  Then $S$
cuts at least $\beta k_i(k_i-1)/2$ edges from $G_i$; in particular if
$\beta = 1/2$ then $S$ cuts at least $k_i(k_i-1)/4$ edges. Conversely if
$S$ cuts at least $\alpha k_i(k_i-1)$ edges for some $\alpha < 1/2$ then no
cell in the arrangement induced by $S$ has more than
$(1 - \Omega(\alpha))k_i$ points from $P_i$.  Given this we can
consider the formulation below.

\begin{eqnarray*}
  \min_{S \subseteq N} |S|  & & \text{s.t}\\
   f_i(S) & \ge &  k_i(k_i-1)/4 \quad 1 \le i \le m
\end{eqnarray*}

We apply our bicriteria approximation for \multisubmodsc with some
fixed $\eps$ to obtain an $O(\log m)$-approximation to the objective
but we are only guaranteed that the output $S$ satisfies the property
that $f_i(S) \ge (1-1/e-\eps)k_i(k_i-1)/4$ for each $i$. This is sufficient
to ensure that no $P_i$ has more than a constant factor in each cell
of  the arrangement.

The running time of the algorithm depends polynomially on $N$ and $m$
and $N$ can be upper bounded as $n^{d}$. The running time in
\cite{HJ18} is $O(mn^{d+2})$. Finding a running time that depends
polynomially on $n,m$ and $d$ is an interesting open  problem.

\section{Sparsity in \partitionsc}
\label{sec:partition}
In this section we consider a problem that generalizes
\partitionsc and \cips while being a special case of \multisubmodsc.
We call this problem \ccf (Covering Coverage Functions).
Bera \etal \cite{BeraGKR} already considered this version in the restricted context
of \vc. Formally the
input is a weighted set system $(\cU, \cS)$ and a set of inequalities
of the form $Az \ge b$ where $A \in [0,1]^{h \times n}$ matrix and
$b \in \mathbb{R}_+^h$ is a positive vector. The goal is to optimize
the integer program CCF-IP shown in Fig~\ref{fig:ccfip}. \partitionsc
is a special case of \ccf when the matrix $A$ contains only $\{0,1\}$
entries. On the other hand \cip is a special case when the set system
is very restricted and each set $S_i$ consists of a single element.
We say that an instance is $r$-sparse if each set $S_i$ ``influences''
at most $r$ rows of $A$; in other words the elements of $S_i$ have
non-zero coefficients in at most $r$ rows of $A$. This notion of
sparsity coincides in the case of \cips with column sparsity and in
the case of \multisubmodsc with the sparsity that we saw in
Section~\ref{sec:bicriteria}.  It is useful to explicitly see
why \ccf is a special case of \multisubmodsc.  The ground set
$N = [m]$ corresponds to the sets $S_1,\ldots,S_m$ in the given
set system $(\cU,\cS)$.  Consider the row $k$ of the covering
constraint matrix $Az \ge b$. We can model it as a constraint
$f_k(S) \ge b_k$ where the submodular set function
$f_k:2^N \rightarrow \mathbb{R}_+$ is defined as follows: for a set
$X \subseteq N$ we let
$f_k(X) = \sum_{e_j \in \cup_{i \in X} S_i} A_{k,j}$ which is
simply a weighted coverage function with the weights coming from the
coefficients of the matrix $A$. Note that when formulating via these
submodular functions, the auxiliary variables $z_1,\ldots,z_n$ that
correspond to the elements $\cU$ are unnecessary.

We prove the following theorem.

\begin{theorem}
  Consider an instance of $r$-sparse \ccf induced by a set system
  $(\cU, \cS)$ from a deletion-closed family with a
  $\beta$-apprximation for \setcover via the natural LP. There is
  a randomized polynomial-time algorithm that outputs a feasible
  solution of expected cost $(\beta + \ln r) \opt$.
\end{theorem}

 \begin{figure}[htb]
  \begin{subfigure}{0.5\linewidth}
    \begin{center}
      \begin{tcolorbox}[width=3in]
        (CCF-IP)
        \begin{align*}
          \min \sum_{S_i \in \cS} w_i x_i & \\
          \sum_{i: e_j \in S_i} x_i & \ge z_j \quad e_j \in \cU\\
          Az & \ge  b \\
          z_j & \in \{0,1\} \quad e_j \in \cU \\
          x_i & \in \{0,1\}  \quad \quad S_i \in \cS
      \end{align*}
    \end{tcolorbox}
  \end{center}
  \caption{Natural IP for \partitionsc. }
  \label{fig:ccfip}
  \end{subfigure}
  \begin{subfigure}{0.5\linewidth}
    \begin{center}
      \begin{tcolorbox}[width=3in]
        (CCF-LP)
        \begin{align*}
          \min \sum_{S_i \in \cS} w_i x_i & \\
          \sum_{i: e_j \in S_i} x_i & \ge z_j \quad e_j \in \cU\\
          A z & \ge b \\
          z_j & \in [0,1] \quad e_j \in \cU \\
          x_i & \in [0,1] \quad \quad S_i \in \cS
      \end{align*}
    \end{tcolorbox}
  \end{center}
  \caption{Natural LP relaxation for CCF-IP. }
  \label{fig:ccflp}
  \end{subfigure}
\end{figure}

The natural LP relaxation for \ccf is show in Fig~\ref{fig:ccflp}.  It
is well-known that this LP relaxation, even for \cips and with one
constraint, has an unbounded integrality gap \cite{CFLP}.  For \cips
knapsack-cover inequalities are used to strengthen the LP.
KC-inequalities in this context were first introduced in the
influential work of Carr \etal \cite{CFLP} and have since become a
standard tool in developing stronger LP relaxations. Bera \etal
\cite{BeraGKR} and
Inamdar and Varadarajan \cite{IV2} adapt KC-inequalities to the setting
of \partitionsc, and it is straight forward to extend this
to \ccf (this is implicit in \cite{BeraGKR}).

\begin{remark}
  Weighted coverage functions are a special case of sums of weighted
  rank functions of matroids. The natural LP for \ccf can be viewed as
  using a different, and in fact a tighter extension, than the
  multilinear relaxation \cite{CCPV07}. The fact that one can use an
  LP relaxation here is crucial to the scaling idea that will play a
  role in the eventual algorithm. The main difficulty, however, is the
  large integrality gap which arises due to the partial covering
  constraints.
\end{remark}

We set up and the explain the notation to describe the use of
KC-inequalities for \ccf. It is convenient here to use the reduction
of \ccf to \multisubmodsc. For row $k$ in $Ax \ge b$ we will use $f_k$
to denote the submodular function that we set up earlier. Recall that
$f_k(D)$ captures the coverage to constraint $k$ if set $D$ is
chosen. The residual requirement after choosing $D$ is $b_k -
f_k(D)$. The residual requirement must be covered by elements
from sets outside $D$.
The maximum contribution that $i \not \in D$ can provide to this
is $\min\{f_k(i), b_k - f_k(D)\}$.  Hence the following constraint is
valid for any $D \subset N$:
$$ \sum_{i \not \in D} \min\{f_k(D+i) - f_k(D), b_k - f_k(D)\} x_i \ge b_k - f_k(D).$$
Writing the preceding inequality for every possible choice of $D$ and
for every $k$ we obtained a strengthened LP that we show in
Fig~\ref{fig:ccflp-kc}.

\begin{figure}[htb]
  \centering
      \begin{tcolorbox}[width=6in]
        (CCF-KC-LP)
        \begin{align*}
          \min \sum_{S_i \in \cS} w_i x_i & \\
          \sum_{i: e_j \in S_i} x_i & \ge z_j \quad e_j \in \cU\\
               A z & \ge b \\
          \sum_{i \not \in D} \min\{f_k(D+i)-f_k(D), b_k - f_k(D)\} x_i & \ge
                                                                 b_k -
                                                                 f_k(D)
                                                                 \quad
          D\subset [m], 1 \le k \le h\\
          z_j & \in [0,1] \quad e_j \in \cU \\
          x_i & \in [0,1] \quad \quad S_i \in \cS
      \end{align*}
    \end{tcolorbox}

  \caption{CCF-LP with KC-Inequalities}
  \label{fig:ccflp-kc}
\end{figure}

CCF-KC-LP has an exponential number of constraints and the separation
problem involves submodular functions. Apriori it is not clear that
there is even an approximate separation oracle. However, one can
combine rounding and separation, as shown in \cite{BeraGKR,IV2}, and
we follow the same approach. The main change is that we use randomized
rounding followed by alteration to fix the uncovered constraints.
This allows us to generalize the approximation ratio to the sparse
case.

We believe that it is instructive to first see how to round the LP
assuming that it can be solved exactly. This assumption
can be avoided as shown in previous work since the rounding requires
only some limited properties from the LP solution.

\paragraph{Rounding and analysis assuming LP can be solved exactly:}
Let $(x,z)$ be an optimum solution to CCF-KC-LP. We can assume without
loss of generality that for each element $e_j \in \cU$ we have
$z_j = \min\{1, \sum_{i: e_j \in S_i} x_i\}$. As in
Section~\ref{sec:psc} we split the elements in $\cU$ into heavily
covered elements and shallow elements. For some fixed threshold $\tau$
that we will specify later, let
$\cU_{\text{he}} = \{ e_j \mid z_j \ge \tau\}$, and
$\cU_{\text{sh}} = \cU \setminus \cU_{\text{he}}$.  We will also
choose another threshold. The rounding algorithm is the following.

\begin{enumerate}
\item Solve a \setcover problem via the natural LP to cover all elements in
  $\cU_{\text{he}}$. Let $Y_1$ be the sets chosen in this step.
\item Let $Y_2 = \{ S_i \mid x_i \ge \tau\}$ be the heavy sets.
\item Repeat for $\ell = \Theta(\ln r)$ rounds: independently pick
  each set $S_i$ in $\cS \setminus (Y_1 \cup Y_2)$ with probability
  $\frac{1}{\tau}x_i$.  Let $Y_3$ be the sets chosen in this
  randomized rounding step.
\item For $k \in [h]$ do
  \begin{enumerate}
  \item Let $b_k - f_k(Y_1 \cup Y_2 \cup Y_3)$ be the residual requirement of
    $k$'th constraint.
  \item Run the \emph{modified} Greedy algorithm to satisfy the
    residual requirement.  Let $F_k$ be the sets chosen to fix the
    constraint (could be empty).
  \end{enumerate}
\item Output $Y_1 \cup Y_2 \cup Y_3 \cup (\cup_{k=1}^h F_k)$.
\end{enumerate}

The algorithm is similar to that in \cite{BeraGKR,IV2}; the main
difference is that we explicitly fix the constraints after the
randomized rounding phase using a slight variant of the Greedy
algorithm. This ensures that the output of the algorithm is always a
feasible solution; this makes it easy to analyze the $r$-sparse case
easily while a straight forward union bound will not work. We now
describe the modified Greedy algorithm to fix a constraint. For an
unsatisfied constraint $k$ we consider the collection of sets that
influence the residual requirement for $k$, and partition them it into
$H_{k}$ and $L_k$.  $H_{k}$ is the collection of all sets such that
choosing any of them completely satisfies the residual requirement for
$k$, and $L_k$ are the remaining sets. The modified Greedy algorithm
for fixing constraint $k$ picks the better of two solutions: (i) the
first solution is the cheapest set in $H_k$ (this makes sense only if
$H_k \neq \emptyset$) and (ii) the second solution is obtained by
running Greedy on sets in $L_k$ until the constraint is satisfied.

\paragraph{Analysis:} We now analyze the expected cost of the solution
output by the algorithm.  Since the high-level ideas are quite similar
to prior work and in the preceding sections, we will sketch the
analysis and focus on a key lemma; it's proof follows from previous
work \cite{BeraGKR,IV2} but we reinterpret it here through
submodularity.

The lemma below bounds the cost of $Y_1$ and its
proof is essentially the same as that of
Lemma~\ref{lem:highlycovered}.

\begin{lemma}
  The cost of $Y_1$, $w(Y_1)$ is at most $\beta \frac{1}{\tau}\sum_i w_i x_i$.
\end{lemma}

The expected cost of randomized rounding in the second step is easy to bound.
\begin{lemma}
  The expected cost of $Y_2$ is at most $\frac{\ell}{\tau} \sum_{i} w_i x_i$.
\end{lemma}

The key technical lemma is the following.
\begin{lemma}[\cite{BeraGKR,IV2}]
  \label{lem:main2}
  Fix a constraint $k$. If $\tau$ is a sufficiently small but fixed
  constant, the probability that constraint $k$ is satisfied after one
  round of randomized rounding is at least a fixed constant $c_\tau$.
\end{lemma}

We will give a different perspective on the preceding lemma in
a paragraph below.
Before that, we finish the rest of the analysis first.

Let $I_k =\{ i \mid \mbox{$S_i$ influences constraint $k$}\}$.
Note that $|I_k| \le r$ by our sparsity assumption.
\begin{lemma}
  \label{lem:fix}
  Let $\rho_k$ be the cost of fixing constraint $k$ if it is not
  satisfied after randomized rounding. Then $\rho_k \le c'_\tau \sum_{i \in
    I_k} w_i x_i$ for some constant $c'_\tau$.
\end{lemma}
\begin{proof}
  We will assume that $\tau < (1-1/e)/2$.  Let $D = Y_1 \cup Y_2$
  and let $b'_k = b_k - f_k(D)$ be residual requirement of constraint
  $k$ after choosing $Y_1$ and $Y_2$.  Let
  $\cU' = \cU \setminus \cU_D$ be elements in the residual instance;
  all these are shallow elements. Consider the scaled solution $x'$
  where $x'_i = 1$ if $S_i \in D$ and $x'_i = \frac{1}{\tau} x_i$ for
  other sets. For any shallow element $e_j$ let
  $z'_j = \min\{1,\sum_{i: j \in S_i} x'_i\}$; since $e_j$ is shallow
  we have $z'_j = \frac{1}{\tau}z_j = \sum_{i: j \in S_i, i \not \in
    D} x'_i$.

  Recall from the description of the modified Greedy algorithm that a
  set $S_i$ is in $H_k \subseteq I_k$ iff adding $S_i$ to $D$
  satisfies constraint $k$. In other words $i \in H_k$ iff
  $f_k(D+i) - f_k(D) \ge b'_k$.  Suppose
  $\sum_{i \in H_k} x'_i \ge 1/2$.  Then it is not hard to see that
  the cheapest set from $H_k$ will cover the residual requirement and
  has cost at most $2 \sum_{i \in H_k} w_i x'_i$ and we are done.  We
  now consider the case when $\sum_{i \in H_k} x'_i < 1/2$.  Let
  $L_k = I_k \setminus H_k$. For each $j \in \cU'$ let
  $z''_j = \sum_{i: j \in S_i, i \in L_k} x'_i$.
 We claim that
  $\sum_{j \in \cU'}A_{k,j}z''_j \ge \frac{1}{2 \tau} b'_k$. Since
  $\tau \le (1-1/e)/2$ this implies
  $\sum_{j \in \cU'}A_{k,j}z''_j \ge \frac{1}{(1-1/e)} b'_k$.
  Assuming the claim, if we run Greedy on $L_k$ to cover at least
  $b'_k$ elements then the total cost, by Lemma~\ref{lem:greedy}, is at most
  $(1+e)\sum_{i \in L_k} w_i x'_i$; note that we use the fact that no
  set in $L_k$ has coverage more than $b'_k$ and hence $c=1$ in
  applying Lemma~\ref{lem:greedy}.

  We now prove the claim. Since the $x,z$ satisfy KC
  inequalities:
  $$ \sum_{i \not \in D, i \in I_k} \min\{f_k(D+i) - f_k(D), b'_k\} x_i \ge b'_k.$$
  We split the LHS into two terms based on sets in $H_k$ and $L_k$.
  Note that if $i \in H_k$ then $f_k(D+i) -f_k(D) \ge b'_k$ and if
  $i \in L_k$ then $f_k(D+i) -f_k(D) < b'_k$. Furthermore,
  $f_k(D+i) -f_k(D) \leq \sum_{e_j \in S_i} A_{k,j}$.
  We thus have
  \begin{eqnarray*}
    \sum_{i \not \in D, i \in I_k} \min\{f_k(D+i) - f_k(D), b'_k\} x_i
    & \leq & \sum_{i \in H_k} b'_k x_i + \sum_{i \in L_k} x_i \sum_{e_j \in
          S_i} A_{k,j} \\
    & \leq & b'_k \sum_{i \in H_k} x_i + \sum_{i \in L_k} x_i \sum_{e_j \in
          S_i} A_{k,j}
  \end{eqnarray*}
  Putting together the preceding two inequalities
  and condition that $\sum_{i \in H_k} x'_i < 1/2$ (recall that
  $x'_i  = x_i/\tau$ for each $i \in I_k$),
  $$\sum_{i \in L_k} x'_i \sum_{e_j \in  S_i} A_{k,j} \ge \frac{1}{2\tau} b'_k.$$
  We have, by swapping the order of summation,
  $$\sum_{i \in L_k} x'_i \sum_{e_j \in  S_i} A_{k,j} = \sum_{e_j \in
    \cup_{i \in L_k}S_i} A_{k,j} \sum_{i \in L_k: e_j \in S_i} x'_i \le
    \sum_{j \in \cU'} A_{k,j} \sum_{i \in L_k: e_j \in S_i} x'_i = \sum_{j \in
      \cU'} A_{k,j} z''_j.$$
    The preceding two inequalities prove the claim.
\end{proof}

With the preceding lemmas we can finish the analysis of the total
expected cost of the sets output by the algorithm. From
Lemma~\ref{lem:main2} the probability that any fixed constraint $k$ is
not satisfied after the randomized rounding step is $c^{-\ell}$. By
choosing $\ell \ge 1+ \log_c r$ we can reduce this probability to at
most $1/r$. Thus, as in the preceding section, the expected fixing
cost is $\sum_{k} \frac{1}{r} w(F_k)$. From Lemma~\ref{lem:fix},
$$\sum_k w(F_k) \le c' \sum_k \sum_{i \in I_k} w_i x_i \le c' r
\sum_i w_i x_i$$ since the given instances is $r$-sparse. Thus the
expected fixing cost is at most $c' \sum_i w_i x_i$. The cost of $Y_1$
is $O(\beta) \sum_i w_i x_i$, the cost of $Y_2$ is
$O(1) \sum_i w_ix_i$, and the expected cost of $Y_3$ is
$O(\log r) \sum_i w_i x_i$.  Putting together, the total expected cost
is at most $O(\beta + \log r) \sum_i w_i x_i$ where the constants
depend on $\tau$. We need to choose $\tau$ to be sufficiently small to
ensure that Lemma~\ref{lem:main2} holds. We do not attempt to
optimize the constants or specify them here.

\paragraph{Submodularity and proof of Lemma~\ref{lem:main2}:}
We follow some notation that we used in the proof of
Lemma~\ref{lem:fix}. Let $D = Y_1 \cup Y_2$ and consider the residual
instance obtained by removing the elements covered by $D$ and reducing
the coverage requirement of each constraint.  The lemma is essentially
only about the residual instance.  Fix a constraint $k$
and recall that $b'_k$ is the residual coverage requirement
and that each set in $H_k$ fully satisfies the requirement by itself.
Recall that $x'_i = \frac{1}{\tau}x_i \le 1$ for each set $i \not \in D$
and $z'_j = \frac{1}{\tau} z_j = \sum_{i: e_j \in S_i} x'_i$
for each residual element $e_j$.
As in the proof of Lemma~\ref{lem:fix} we consider two cases.
If $\sum_{i \in H_k} x'_i \ge 1/2$ then with probability
$(1-1/\sqrt{e})$ at least one set from $H_k$ is picked and will
satisfy the requirement by itself. Thus the interesting case
is when  $\sum_{i \in H_k} x'_i < 1/2$.
Let $\cU'' = \cup_{i \in L_k} S_i$.
As we saw earlier,
in this case
$$\sum_{j \in \cU''} A_{k,j} \min\{1,\sum_{i: j \in S_i} x'_i\} \ge
\frac{1}{2\tau}b'_k.$$ For ease of notation we let $N = L_k$ be a
ground set.  Consider the weighted coverage function
$g: 2^N \rightarrow \mathbb{R}_+$ where $g(T)$ for $T \subseteq L_k$
is given by $\sum_{j \in \cup_{i \in T} S_i} A_{k,j}$. Then for a
vector $y \in [0,1]^N$ the quantity
$\sum_{j \in \cU''} A_{k,j} \min\{1,\sum_{i: j \in S_i} y_i\}$ is the
continuous extension $\tilde{g}(y)$ discussed in
Section~\ref{sec:prelim}.  Thus we have
$\tilde{g}(x') \ge \frac{1}{2\tau} b'_k$.  From
Lemma~\ref{lem:extensions}, we have
$G(x') \ge (1-1/e)\frac{1}{2\tau} b'_k$ where $G$ is the multilinear
extension of $g$.  If we choose $\tau \le (1-1/e)/4$ then
$G(x') \ge 2b'_k$.  Let $Z$ be the random variable denoting the value
of $g(R)$ where $R \simeq x'$. Independent random rounding of $x'$
preserves $G(x')$ in expectation by the definition of the multilinear
extension, therefore $\E[Z] = G(x') \geq 2b'_k$.  Moreover, by
Lemma~\ref{lem:submod-concentration}, $Z$ is concentrated around its
expectation since $G(i) \le b'_k$ for each $i \in L_k$. An easy
calculation shows that $\Pr[Z < b'_k] \le e^{1/4} < 0.78$. Thus
with constant probability $g(R) \ge b'_k$.

\paragraph{Solving the LP with KC inequalities}
As noted in prior work \cite{BeraGKR,IV2}, one can combine the
rounding procedure with the Ellipsoid method to obtain the desired
guarantees even though we do not obtain a fractional solution that
satisfies all the KC inequalities. This observation holds for our
rounding as well. We briefly sketch the argument.

The proof of the performance guarantee of the algorithm relies on the
fractional solution satisfying KC inequalities with respect to the set
$D = Y_1 \cup Y_2$. Thus, given a fractional solution $(x,z)$ for the
LP we can check the easy constraints in polynomial time and implement
the first two steps of the algorithm. Once $Y_1,Y_2$ are determined we
have $D$ and one can check if $(x,z)$ satisfies KC inequalities with
respect to $D$ (for each row of $A$). If it does then the rest of the
proof goes through and performance guarantee holds with respect to the
cost of $(x,z)$ which is a lower bound on $\opt$.  If some constraint
does not satisfy the KC inequality with respect to $D$ we can use this
as a separation oracle in the Ellipsoid method.

\section{Concluding Remarks}
The paper shows the utility of viewing \psc and its generalizations as
special cases of \multisubmodsc.  The coverage function in set systems
is a submodular funtion that belongs to the class of sum of weighted
matroid rank functions. Certain ideas for the coverage function extend
to this larger class. Are there interesting problems that can be
understood through this view point? Are there other special classes of
submodular functions for which one can obtain uni-criteria
approximation algorithms for \multisubmodsc unlike the bicriteria one
we presented? An interesting example is the problem considered in
\cite{HJ18}. The algorithm in this paper for \partitionsc, like the
ones in \cite{BeraGKR,IV2}, relies on using the Ellipsoid method to
solve the LP with KC inequalities. It may be possible to avoid
the inherent inefficiency in this way of solving the LP via some
ideas from recent and past work \cite{CFLP,CQ19}.

\paragraph{Acknowledgements:} CC thanks Sariel Har-Peled, Tanmay
Inamdar and Kasturi Varadarajan for discussion and comments.

\bibliographystyle{alpha}
\bibliography{partial-set-cover}

\begin{thebibliography}{BGKR14}

\bibitem[AS04]{AgeevS}
Alexander~A Ageev and Maxim~I Sviridenko.
\newblock Pipage rounding: A new method of constructing algorithms with proven
  performance guarantee.
\newblock {\em Journal of Combinatorial Optimization}, 8(3):307--328, 2004.

\bibitem[BGKR14]{BeraGKR}
Suman~K Bera, Shalmoli Gupta, Amit Kumar, and Sambuddha Roy.
\newblock Approximation algorithms for the partition vertex cover problem.
\newblock {\em Theoretical Computer Science}, 555:2--8, 2014.

\bibitem[CCPV07]{CCPV07}
Gruia Calinescu, Chandra Chekuri, Martin P{\'a}l, and Jan Vondr{\'a}k.
\newblock Maximizing a submodular set function subject to a matroid constraint
  (extended abstract).
\newblock In {\em Integer Programming and Combinatorial Optimization (IPCO)},
  pages 182--196, 2007.

\bibitem[CCPV11]{CCPV}
Gruia C{\u{a}}linescu, Chandra Chekuri, Martin P{\'a}l, and Jan Vondr{\'a}k.
\newblock Maximizing a monotone submodular function subject to a matroid
  constraint.
\newblock {\em SIAM J. Comput.}, 40(6):1740--1766, 2011.

\bibitem[CFLP00]{CFLP}
Robert~D. Carr, Lisa Fleischer, Vitus~J. Leung, and Cynthia~A. Phillips.
\newblock Strengthening integrality gaps for capacitated network design and
  covering problems.
\newblock In {\em Proceedings of {ACM-SIAM} SODA}, pages 106--115, 2000.

\bibitem[CGKS12]{ChanGKS}
Timothy~M Chan, Elyot Grant, Jochen K{\"o}nemann, and Malcolm Sharpe.
\newblock Weighted capacitated, priority, and geometric set cover via improved
  quasi-uniform sampling.
\newblock In {\em Proceedings of the twenty-third annual ACM-SIAM symposium on
  Discrete Algorithms}, pages 1576--1585. Society for Industrial and Applied
  Mathematics, 2012.

\bibitem[CHS16]{CHS}
Antares Chen, David~G Harris, and Aravind Srinivasan.
\newblock Partial resampling to approximate covering integer programs.
\newblock In {\em Proceedings of 27th ACM-SIAM SODA}, pages 1984--2003, 2016.

\bibitem[CJV15]{CJV}
Chandra Chekuri, T.S. Jayram, and Jan Vondr\'ak.
\newblock On multiplicative weight updates for concave and submodular function
  maximization.
\newblock In {\em Proceedings of ITCS}, 2015.

\bibitem[CQ19]{CQ19}
Chandra Chekuri and Kent Quanrud.
\newblock On approximating (sparse) covering integer programs.
\newblock In {\em Proceedings of the Thirtieth Annual ACM-SIAM Symposium on
  Discrete Algorithms}, pages 1596--1615. SIAM, 2019.

\bibitem[CVZ10]{CVZ}
Chandra Chekuri, Jan Vondr{\'{a}}k, and Rico Zenklusen.
\newblock Dependent randomized rounding via exchange properties of
  combinatorial structures.
\newblock In {\em 51th Annual {IEEE} Symposium on Foundations of Computer
  Science, {FOCS} 2010, October 23-26, 2010, Las Vegas, Nevada, {USA}}, pages
  575--584, 2010.

\bibitem[Fei98]{Feige}
Uriel Feige.
\newblock A threshold of $\ln n$ for approximating set cover.
\newblock {\em J. {ACM}}, 45(4):634--652, 1998.
\newblock Preliminary version in STOC 1996.

\bibitem[HJ18]{HJ18}
Sariel Har{-}Peled and Mitchell Jones.
\newblock Few cuts meet many point sets.
\newblock {\em CoRR}, abs/1808.03260, 2018.

\bibitem[HK18]{hk18}
Eunpyeong Hong and Mong-Jen Kao.
\newblock {Approximation Algorithm for Vertex Cover with Multiple Covering
  Constraints}.
\newblock In Wen-Lian Hsu, Der-Tsai Lee, and Chung-Shou Liao, editors, {\em
  29th International Symposium on Algorithms and Computation (ISAAC 2018)},
  volume 123 of {\em Leibniz International Proceedings in Informatics
  (LIPIcs)}, pages 43:1--43:11, Dagstuhl, Germany, 2018. Schloss
  Dagstuhl--Leibniz-Zentrum fuer Informatik.

\bibitem[IV18a]{IV1}
Tanmay Inamdar and Kasturi~R. Varadarajan.
\newblock On partial covering for geometric set systems.
\newblock In {\em 34th International Symposium on Computational Geometry, SoCG
  2018, June 11-14, 2018, Budapest, Hungary}, pages 47:1--47:14, 2018.

\bibitem[IV18b]{IV2}
Tanmay Inamdar and Kasturi~R. Varadarajan.
\newblock On the partition set cover problem.
\newblock {\em CoRR}, abs/1809.06506, 2018.

\bibitem[KMN99]{KMN}
Samir Khuller, Anna Moss, and Joseph~Seffi Naor.
\newblock The budgeted maximum coverage problem.
\newblock {\em Information processing letters}, 70(1):39--45, 1999.

\bibitem[KPS11]{KPS}
Jochen K{\"o}nemann, Ojas Parekh, and Danny Segev.
\newblock A unified approach to approximating partial covering problems.
\newblock {\em Algorithmica}, 59(4):489--509, 2011.

\bibitem[KY05]{KY}
Stavros~G. Kolliopoulos and Neal~E. Young.
\newblock Approximation algorithms for covering/packing integer programs.
\newblock {\em J. Comput. Syst. Sci.}, 71(4):495--505, 2005.
\newblock Preliminary version in FOCS 2001.

\bibitem[MR10]{MR}
Nabil~H Mustafa and Saurabh Ray.
\newblock Improved results on geometric hitting set problems.
\newblock {\em Discrete \& Computational Geometry}, 44(4):883--895, 2010.

\bibitem[MRR15]{MRR}
Nabil~H Mustafa, Rajiv Raman, and Saurabh Ray.
\newblock Quasi-polynomial time approximation scheme for weighted geometric set
  cover on pseudodisks and halfspaces.
\newblock {\em SIAM Journal on Computing}, 44(6):1650--1669, 2015.

\bibitem[NWF78]{NWF}
George~L Nemhauser, Laurence~A Wolsey, and Marshall~L Fisher.
\newblock An analysis of approximations for maximizing submodular set
  functions—i.
\newblock {\em Mathematical Programming}, 14(1):265--294, 1978.

\bibitem[Svi04]{Svir}
Maxim Sviridenko.
\newblock A note on maximizing a submodular set function subject to a knapsack
  constraint.
\newblock {\em Operations Research Letters}, 32(1):41--43, 2004.

\bibitem[Von07]{Vondrak-thesis}
Jan Vondr\'ak.
\newblock {\em Submodularity in combinatorial optimization}.
\newblock PhD thesis, Charles University, 2007.
\newblock Avaulable at
  \url{https://theory.stanford.edu/~jvondrak/data/KAM_thesis.pdf}.

\bibitem[Von08]{V08}
Jan Vondr{\'a}k.
\newblock Optimal approximation for the submodular welfare problem in the value
  oracle model.
\newblock In {\em Proceedings of the fortieth annual ACM symposium on Theory of
  computing}, pages 67--74. ACM, 2008.

\bibitem[Von10]{submod-chernoff}
Jan Vondr{\'{a}}k.
\newblock A note on concentration of submodular functions.
\newblock {\em CoRR}, abs/1005.2791, 2010.

\bibitem[Wol82]{wolsey}
Laurence~A Wolsey.
\newblock An analysis of the greedy algorithm for the submodular set covering
  problem.
\newblock {\em Combinatorica}, 2(4):385--393, 1982.

\end{thebibliography}
\end{document}